\documentclass[conference]{IEEEtran}
\IEEEoverridecommandlockouts
\usepackage{cite}
\usepackage{amsmath,amssymb,amsfonts}
\usepackage{graphicx}
\usepackage{textcomp}
\usepackage{xcolor}

\usepackage{graphicx}
\usepackage{algpseudocode,algorithm}
\usepackage{bm}
\usepackage{url}
\usepackage{subfigure}
\usepackage{hhline}
\usepackage{hyperref}
\usepackage{textcomp}
\usepackage{xcolor}
\usepackage{comment}
\usepackage[hang,flushmargin]{footmisc}
\usepackage{lipsum}
\makeatletter
\newcommand{\algorithmfootnote}[2][\footnotesize]{%
  \let\old@algocf@finish\@algocf@finish
  \def\@algocf@finish{\old@algocf@finish
    \leavevmode\rlap{\begin{minipage}{\linewidth}
    #1#2
    \end{minipage}}%
  }%
}
\makeatother

\newcommand{\argmax}{\mathop{\rm arg~max}\limits}

\newcommand{\sign}{\mathop{\rm sign}\limits}
\usepackage{amsthm}
\theoremstyle{definition}

\newtheorem{problem}{Problem}
\newtheorem{lemma}{Lemma}
\newtheorem{proposition}{Proposition}
\newtheorem{definition}{Definition}
\usepackage{xparse}

\NewDocumentCommand{\timeseries}{O{}O{}O{}}{#1_{#2}[{#3}]}
\NewDocumentCommand{\ret}{O{}O{}}{\timeseries[r][#1][#2]}
\NewDocumentCommand{\vecret}{O{}O{}}{\bm{r}_{#1}[{#2}]}
\NewDocumentCommand{\predret}{O{}O{}}{\hat{r}_{#1}[{#2}]}
\NewDocumentCommand{\signret}{O{}O{}}{\hat{b}_{#1}[{#2}]}
\NewDocumentCommand{\price}{O{}O{}}{X_{#1}[{#2}]}
\NewDocumentCommand{\numterm}{}{M}
\newcommand{\mean}{\mathop{\rm mean}\limits}

\usepackage{subfigure}
\makeatletter
\newcommand{\figcaption}[1]{\def\@captype{figure}\caption{#1}}
\newcommand{\tblcaption}[1]{\def\@captype{table}\caption{#1}}
\makeatother

\def\BibTeX{{\rm B\kern-.05em{\sc i\kern-.025em b}\kern-.08em
    T\kern-.1667em\lower.7ex\hbox{E}\kern-.125emX}}
\begin{document}

\title{Uncertainty Aware Trader-Company Method: Interpretable Stock Price Prediction Capturing Uncertainty}

\author{\IEEEauthorblockN{1\textsuperscript{st} Yugo Fujimoto}
\IEEEauthorblockA{\textit{Innovation Lab} \\
\textit{Nomura Asset Management Co, Ltd.}\\
Tokyo, Japan \\
yu5fujimoto@gmail.com}
\and
\IEEEauthorblockN{2\textsuperscript{nd} Kei Nakagawa}
\IEEEauthorblockA{\textit{Innovation Lab} \\
\textit{Nomura Asset Management Co, Ltd.}\\
Tokyo, Japan \\
kei.nak.0315@gmail.com}
\and
\IEEEauthorblockN{3\textsuperscript{rd} Kentaro Imajo}
\IEEEauthorblockA{
\textit{Preferred Networks, Inc.}\\
Tokyo, Japan \\
imos@preferred.jp}
\and
\IEEEauthorblockN{4\textsuperscript{th} Kentaro Minami}
\IEEEauthorblockA{
\textit{Preferred Networks, Inc.}\\
Tokyo, Japan \\
minami@preferred.jp}
}

\maketitle

\begin{abstract}
Machine learning is an increasingly popular tool with some success in predicting stock prices.
One promising method is the Trader-Company~(TC) method, which takes into account the dynamism of the stock market and has both high predictive power and interpretability.
Machine learning-based stock prediction methods including the TC method have been concentrating on point prediction.
However, point prediction in the absence of uncertainty estimates lacks credibility quantification and raises concerns about safety. 
The challenge in this paper is to make an investment strategy that combines high predictive power and the ability to quantify uncertainty.
We propose a novel approach called Uncertainty Aware Trader-Company Method~(UTC) method. 
The core idea of this approach is to combine the strengths of both frameworks by merging the TC method with the probabilistic modeling, which provides probabilistic predictions and uncertainty estimations.
We expect this to retain the predictive power and interpretability of the TC method while capturing the uncertainty.
We theoretically prove that the proposed method estimates the posterior variance and does not introduce additional biases from the original TC method.
We conduct a comprehensive evaluation of our approach based on the synthetic and real market datasets.
We confirm with synthetic data that the UTC method can detect situations where the uncertainty increases and the prediction is difficult. 
We also confirmed that the UTC method can detect abrupt changes in data generating distributions.
We demonstrate with real market data that the UTC method can achieve higher returns and lower risks than baselines.
\end{abstract}

\begin{IEEEkeywords}
Finance, Metaheuristics, Stock Price Prediction, Uncertainty.
\end{IEEEkeywords}

\section{Introduction} \label{ch:intro}
Stock price predictability has been an important research topic in both academia and industry since it reflects our economic and social organization and the stock market plays an important role in the world economy. 
Although the dynamic nature of our economic activity makes it harder to predict future stock prices, significant efforts are made to explain the dynamism. 
From this perspective, stock markets have often been modeled as a complex, evolutionary, and nonlinear dynamical system~\cite{cont2001empirical,hommes2001financial,borges2010efficient}.

Due to the dynamic nature of our economic activity, machine learning is an increasingly popular tool with some success in predicting stock prices~\cite{henrique2019literature,nakagawa2020ric,imajo2021deep}.
This is because many machine learning methods can automatically capture nonlinear relationships between relevant factors from the input data~\cite{cavalcante2016computational,chen2018integrating}.
One promising method among them is the Trader-Company~(TC) method, a metaheuristic stock prediction model that mimics the roles of an actual financial institute and traders within it~\cite{ito2021trader}.
The TC method consists of two components, the predictor called a Trader and the aggregation algorithm called a Company.
The TC method considers the dynamism of the stock market and has both high predictive power and interpretability.

Stock prediction methods based on machine learning, including the TC method, have been concentrating on estimating and improving point predictions.
However, point predictions in the absence of uncertainty estimates lack credibility quantification and raise concerns about safety.
Considering the significant consequences of decision-making in financial practice, quantifying the uncertainty of predictions is proving to be a key step in putting machine learning models into practice~\cite{jaeger2021understanding,philps2021interpretable}.
For example, most trades (80\%) are automated~\cite{URL} and the algorithmic tradings based on the machine learning method have played a crucial role in financial markets.
The algorithmic tradings focused on the large investment universe of stocks and sampled data at very high frequencies (intraday or tick by tick).
In such an environment with a large amount of data, it is important for practitioners to quantify the uncertainty of predictions.
Therefore, the challenge in this paper is to make an investment strategy with high predictive power and can quantify the uncertainty of predictions.

To formalize our discussion of the uncertainty of predictions, we will rely on probabilistic modeling. Probabilistic modeling, which can provide probabilistic predictions and uncertainty estimations simultaneously, has been a fundamental tool in machine learning and related fields~\cite{bishop2006pattern}. 
Most of these studies rely on a Bayesian framework, and their applications to complex models such as neural networks and decision tree models have been actively studied. Among them, one standard approach is directly estimating the distribution of predictions. However, it has been pointed out that these methods tend to make predictions biased toward one specific mode~\cite{Fort2019,malinin2021uncertainty}.
Another approach is ensemble-based uncertainty estimation, which focuses on the dispersion of predictions~\cite{pmlr-v48-gal16,NIPS2017_9ef2ed4b,malinin2021uncertainty}. 
These methods are experimentally confirmed to be more robust to dataset shift than methods that explicitly learn distributions~\cite{NIPS2017_9ef2ed4b,malinin2021uncertainty}. That is effective in predicting a dynamic environment, such as financial markets.

Based on the above studies, we propose a novel approach called the Uncertainty Aware Trader-Company~(UTC) method.
The core idea of this approach is to combine the strengths of both frameworks by merging the TC method with the 
probabilistic modeling framework. 
We expect to retain the predictive power and interpretability of the TC method while capturing the uncertainty.
To be more concrete, we propose the method of estimating the prediction's uncertainty from the Traders' output.
We estimate the variance by the two-stage algorithm: estimation by the Trader and estimation by the Company.
The Trader estimates the uncertainty on weighting given a trader's strategy while the Company estimates the uncertainty about the effectiveness of the whole traders' strategies.
We theoretically show that our uncertainty estimation reflects the posterior variance of predictive return given the past return.
Also, we prove that the predictive return of the UTC method is identical to that of the original TC method under some assumptions in the prior distribution of traders. That means our method does not introduce additional biases.

We conduct a comprehensive evaluation of our approach based on the synthetic and real market datasets.
Our evaluation of the synthetic datasets demonstrates that the UTC method can detect situations where the uncertainty increases and the prediction is difficult.
We also confirmed that our method can detect abrupt changes of data generating distributions.
Furthermore, experiments using actual data show that the investment strategy based on our UTC method gains stable returns while suppressing risks compared to existing investment strategies.

The remainder of this paper is organized as follows. 
Section 2 describes our problem formulation and TC method briefly. Section 3 presents our UTC method and theoretical properties. Section 4 performs experiments. Section 5 reviews the related work, and Section 6 is the conclusion.

\section{Preliminary} \label{sec:prelim}
In this section, we formulate our problem and then provide the overview of the TC method~\cite{ito2021trader}.

\subsection{Problem Formulation} \label{sec:problem}
Our problem is to forecast future returns of stocks based on their historical observations. 

Let $\price[i][t]$ be the price of stock $i$ at time $t$ where $1 \le i \le S$ denotes the index of stocks and $0 \le t \le T$ denotes the time index.
We use the logarithmic returns of stock prices as input features of models; we denote the one period ahead return of stock $i$ by 
\begin{align}\label{logret}
\ret[i][t] := \log (\price[i][t] / \price[i][t-1]) \approx \frac{\price[i][t]-\price[i][t-1]}{\price[i][t-1]}.
\end{align}

Then we define the returns of stock $i$ and returns of multiple stocks $1 \le i \le j \le S$ from time period $u$ to $v~(u \leq v)$ by
\begin{align}\label{vecret}
&\vecret[i][u:v]:=(\ret[i][u],\cdots,\ret[i][v]),\\
&\vecret[i:j][u:v]:=(\ret[i][u:v],\cdots,\ret[j][u:v])
\end{align}

We can formulate our main problem as follows.
\begin{problem}[one-period-ahead prediction]\label{main_problem}
We sequentially observe the returns $\ret[i][t]$ ($1 \le i \le S$) at every time $0 \le t \le T-1$. 
We predict the one-period-ahead return $\ret[i][t+1]$ and estimate its predictive uncertainty $\sigma_i[t+1]$ based on the past $t$ returns $\vecret[1:S][0:t]$. 
That is, the one-period-ahead return and uncertainty prediction can be written as
\begin{equation}\label{predret}
\predret[i][t+1], \hat{\sigma}_i[t+1] = f_{t}(\vecret[1:S][0:t])
\end{equation}
for some function $f_t$.
The purpose of this study is to find $f_t$ whose output $\predret[i][t+1]$ approximates the true return $\ret[i][t+1]$ and predictive standard deviation $\hat{\sigma}_i[t+1]$ approximates the estimation error between  $\predret[i][t+1]$ and $\ret[i][t+1]$ well.
\end{problem}

\subsection{Trader Company Method}\label{sec:tc}
This section introduces the Trader-Company method briefly. The TC method consists of two main components, \textit{Traders} and \textit{Companies}. A Trader predicts the returns using a simple model expressing realistic trading strategies, while a Company combines strategies from multiple Traders into a single prediction. 
A Company applies an evolutionary algorithm that mimics the role of financial institutes as employers of traders. 
During training, a Company generates promising new candidates for Traders and deletes poorly performing ones.
We provide more detailed definitions and training algorithms for the TC method.

\begin{table}
\caption{Notation.}
\label{tab:notation}
\begin{center}

\begin{tabular}{clc}
\hline
\textbf{Notation} & \multicolumn{1}{c}{\textbf{Meaning}} & \textbf{Def.} \\
\hline
$\price[i][t]$ & 
\begin{tabular}{l}
stock price of stock $i$ at time $t$\\
where $1\le i \le S, 0\le t \le T$
\end{tabular}& $\S$ \ref{sec:problem} \\
$\ret[i][t]$ & logarithmic return of $i$ at $t$& \eqref{logret}\\
$\vecret[i][u:v]$ & $(\ret[i][u],\cdots,\ret[i][v])$ & \eqref{vecret}  \\
$\vecret[i:j][u:v]$ & $(\ret[i][u:v],\cdots,\ret[j][u:v])$ & \eqref{vecret}  \\
$\predret[i][t+1],\hat{\sigma}_i[t+1]$ & predicted value and standard deviation of $\ret[i][t]$& \eqref{predret}\\
\begin{tabular}{l}$\numterm,P_j,Q_j,D_j$\\$F_j,A_j,O_j$\end{tabular}
& hyper-parameters of Traders & \eqref{TCtraders}\\
\hline
\end{tabular}
\end{center}

\end{table}

\subsubsection{Traders - Simple Prediction Module} \label{subsec:TCtraders}
\begin{definition}
A Trader is a predictor of one period ahead returns defined as follows.
Let $\numterm$ be the number of terms in the prediction formula.
For each $1 \le j \le \numterm$, we define $P_j,Q_j$ as the indices of the stock to use, $D_j,F_j$ as the delay parameters, $O_j$ as the binary operator, $A_j$ as the activation function, and $w_j$ as the weight of the $j$-th term. 
Then, the Trader predicts the return value $\ret[i][t+1]$ at time $t + 1$ by the formula
\begin{align}\label{TCtraders}
    f_{\theta,w}(\vecret[1:S][0:t])= \sum_{j=1}^\numterm w_j                 A_j(O_j(\ret[P_j][t-D_j],\ret[Q_j][t-F_j])).
\end{align}
where $\theta$ is the parameters of the Trader:
$$\theta := \{\numterm,\{P_j,Q_j,D_j,F_j,O_j,A_j\}_{j=1}^\numterm\}.$$

\end{definition}
For activation functions $A_j$, the TC method used standard activation functions used in deep learning, such as the identity function, hyperbolic tangent function, hyperbolic sine function, and Rectified Linear Unit (ReLU~\cite{nair2010rectified}). 
For the binary operators $O_j$, we use several arithmetic binary operators (e.g., $x + y$, $x - y$, and $x \times y$), the coordinate projection, $(x, y) \mapsto x$, the max/min functions, and the comparison function $(x > y) = \sign(x - y)$.

The formula \eqref{TCtraders} is interpretable in that it has a similar form to typical human-generated trading strategies ~\cite{Kakushadze2018}.
Second, the Trader model has sufficient expressive power.
The Trader has various binary operators as fundamental units, which allows it to represent any binary operations commonly used in practical trading strategies. Besides, the model also 
encompasses the linear models since we can choose the projection operator $(x, y) \mapsto x$ as $O_j$.

The Trader is optimized to maximize the cumulative return of its strategy.
\begin{align}\label{traderloss}
    \argmax_{\theta,w} \sum_{u}\sign(f_{\theta,w}(\vecret[1:S][0:u]))\cdot \ret[i][u+1] 
\end{align}
Since the parameter $\theta$ is a discrete variable, standard optimization methods with derivatives are difficult to apply.
Therefore, the TC method introduces an evolutionary algorithm driven by Company models.

\subsubsection{Companies - Optimization and Aggregation Module}\label{sec:TCcompany}
\begin{algorithm}[ht]
    \caption{Educate algorithm of Company in TC}
    \algrenewcommand\algorithmicrequire{\textbf{Input:}}
    \algrenewcommand\algorithmicensure{\textbf{Output:}}
    \label{TCcompanyeducate}
    \begin{algorithmic}[1]
    \Require $\vecret[1:S][0:t]$:stock returns before $t$
    \Require Traders. $N$ : the number of Traders. $Q$: ratio.
    \Ensure Traders
    \Function{CompanyEducate}{}
    \State{$R_n\Leftarrow R(f_{\theta_n,w_n},\vecret[1:S][0:t],\ret[i][0:t+1])$}
    \Comment{Trader's return \eqref{traderloss}}

    \State{$R^* \Leftarrow$ bottom $Q$ percentile of $\{R_n \}$}
    \For{$n \in \{m| R_m \le R^* \}$}
        \Comment{for all bad traders}

        \State{Update $w_i$ in (\ref{TCtraders}) by least squares method}
    \EndFor\\
    \Return{Traders}
    \EndFunction
    \end{algorithmic}
\end{algorithm}

\begin{algorithm}[ht]
    \caption{Prune-and-Generate algorithm of Company}
    \algrenewcommand\algorithmicrequire{\textbf{Input:}}
    \algrenewcommand\algorithmicensure{\textbf{Output:}}
    \label{TCcompanygen}
    \begin{algorithmic}[1]
    \Require $\vecret[1:S][0:t]$:stock returns before $t$, F: \# of fit times
    \Require $N$: the number of Predictors. $Q$: ratio.
    \Ensure $N'$ Predictors 
    \State{$\theta_n,w_n \sim $ Uniform Distribution}
    \For{$k=1,\cdots ,F$}
    \State{$R_n \Leftarrow R(f_\Theta,\vecret[1:S][0:t],\ret[i][0:t+1])$}
    \Comment{Trader's return \eqref{traderloss}}

    \State{$R^* \Leftarrow$ bottom $Q$-percentile of $\{ R_n \}$}
    \State{$\{(\theta_j,w_j)\}_j \Leftarrow \{(\theta_n,w_j) | R_n \ge R^* \}$}
    \Comment{Pruning}
    \State{$\{(\theta_j,w_j)\}_{j=1}^{N'} \sim$ GM fitted to $\{(\theta_j,w_j)\}_j$  *}
    \Comment{Generation}
    \EndFor\\
    \Return{$N'$ Predictors with $\{(\theta_j,w_j)\}_{j=1}^{N'}$}
    \end{algorithmic}
    * If the parameter is an integer, we round it off.

\end{algorithm}

In this framework, a Company maintains $N$ Traders that act as weak learners or feature extractors and aggregate them. Given $N$ Traders specified by parameters $\theta_1, \ldots,\theta_N$,
$w_1, \ldots, w_N$ and the past observations of stock returns $\vecret[1:S][0:t]$, a Company predicts the one-period-ahead return by
\[
    \hat{r}[t + 1] =
    \mathrm{Aggregate}(f_{\theta_1,w_1}, \ldots, f_{\theta_n,w_n}).
\]
Here, we employed the simple averaging 
$$\frac{1}{N} \sum_{n = 1}^N f_{\theta_n,w_n}(\vecret[1:S][0:t])$$
for $\mathrm{Aggregate}$ function.

The Company should maintain the average quality as well as the diversity of the Traders' strategies to achieve low training errors whilst avoiding overfitting.
For this purpose, the TC method introduced the Educate algorithm (Algorithm \ref{TCcompanyeducate}) and the Prune-and-Generate algorithm (Algorithm \ref{TCcompanygen}), which update the weights and formulae of Traders, respectively.

\begin{description}
    \item[\textbf{Educating Traders}:]  ~\\
    Recall that a Trader \eqref{TCtraders} is a linear combination of $M$ mathematical formula. 
    We update the weights $\{ w_j \}_j $ by the least-squares method.
    \item[\textbf{Pruning Traders and generating new candidates}:]~\\
    Since the discrete parameters such as stock indices are difficult to optimize, they update these parameters by Algorithm \ref{TCcompanygen}. 
    First, we evaluate the (cumulative) returns of the Traders and remove them with relatively low returns. 
    Then, we generate new Traders by randomly fluctuating the existing Traders with good performances, i.e., we fit continuous Gaussian mixture distribution, draw new parameters from it and discretize the generated samples for discrete indices. 
\end{description}

\begin{figure*}[t]
\begin{center}
    \includegraphics[width=2.0\columnwidth]{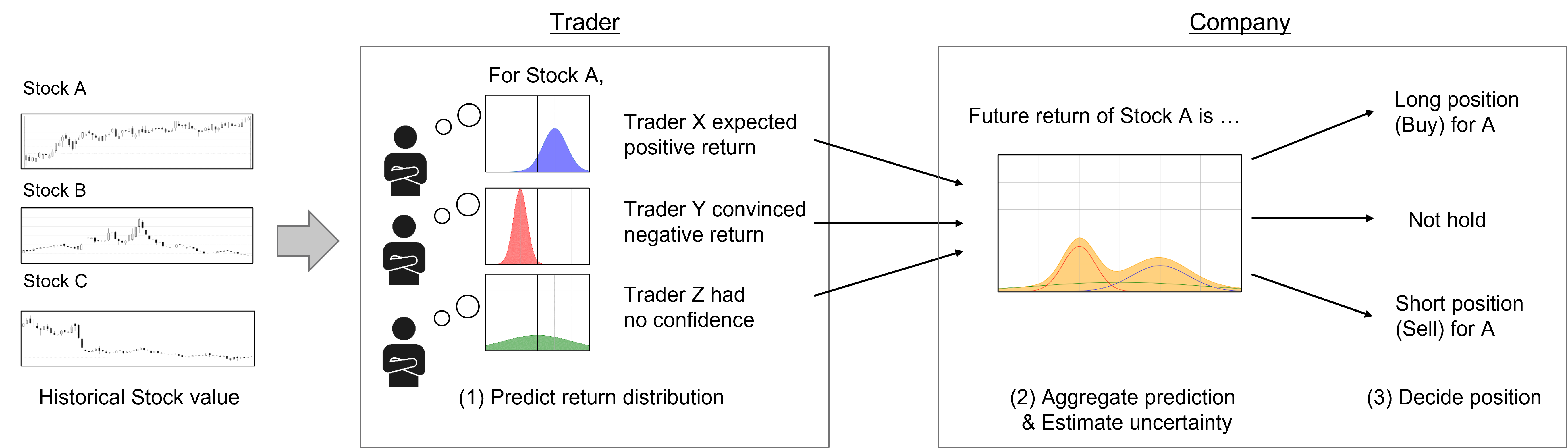}
    \caption{Overview of prediction procedure in our Uncertainty Aware Trader-Company Method}
    \label{fig:UTCprocess}
\end{center}
\end{figure*}
\section{Uncertainty Aware Trader-Company Method}
In this section, we present a \textit{Uncertainty Aware Trader-Company Method}, which extends the Trader-Company method to consider the uncertainty.
As with the original TC method, our method consists of two main components, \textit{Traders} and \textit{Companies}. 
We propose extensions in three major aspects to exploit the uncertainty in the Trader's prediction and the diversification of Traders' prediction.
We show the overview of our method in Figure \ref{fig:UTCprocess}.
\subsection{Trader} \label{sec:trader}
We basically follow the definition of traders in TC methods in section \ref{subsec:TCtraders}.
The difference from the TC method is that we assume each parameter $\theta_n$ and $w_n$ is randomly chosen from the empirical distribution $\Theta$.
As mentioned in the introduction, few solid strategies will always be effective due to the volatile and uncertain behavior of the financial market. It is important to estimate the reliability of the strategies in such an environment. Therefore, instead of learning a deterministic strategy, we learn a probability distribution over strategies to estimate the uncertainty of the prediction.
In other words, our goal is to find a distribution $\Theta$ which approximates the posterior distribution.

Next, we will discuss how the Trader estimates the predictive uncertainty.
Since we assumed the weights of the Trader are chosen from the empirical distribution $\Theta$, the output of the Trader is also probabilistic.
Here, we assumed that the distribution of weights follows a multivariate normal distribution with a mean of $m$ and the covariance of $\Sigma$.
We can calculate the mean and variance for the Trader's output as follows.
\begin{align}\label{btctrader}
    & \mu_n = m^\top z_n,~~\sigma_n^2 = z_n^\top \Sigma z_n
\end{align}
Here, $z_n$ is the signal defined as 
\begin{align} \label{eq:signal}
z_n[j] = A_j(O_j(\ret[P_j][t-D_j],\ret[Q_j][t-F_j])).
\end{align}
The variance of Trader's prediction estimated here is used to predict the return and estimate the variance of the Company prediction.

Next, we will discuss how the Company optimizes the distribution $\Theta$.
\subsection{Company}\label{sec:company}
We basically employ the same framework as the TC method; our Company maintains $N$ traders that act as weak learners or feature extractors and aggregate them. 
However, each Trader estimates not only the expected value but the variance of the predicted return. 
First, we introduce how to estimate the predictive return and its uncertainty based on these outputs.
Given $N$ Traders 
with parameters $\theta_1, \ldots, \theta_n$ and the past observations of stock returns $\vecret[1:S][0:t]$, a Company predicts the one-period-ahead return $\hat{r}[t + 1]$.
The original TC method applies the most naive aggregating method, simply averaging.
If we apply this, we could calculate the mean and variance of the prediction by
\begin{align}
    &\mu = \frac{1}{N}\sum_{n \in \{1,2,\cdots,N\}} \mu_n,~~\sigma = \sqrt{\sigma^2_{intra}+\sigma^2_{inter}}, \nonumber\\
    & \sigma^2_{intra} = \frac{1}{N}\sum_{n \in \{1,2,\cdots,N\}} (\mu_n - \mu)^2,\nonumber \\ & \sigma^2_{inter} = \frac{1}{N}\sum_{n \in \{1,2,\cdots,N\}} \sigma_n^2. \label{eq:variance}
\end{align}

For clarity, this procedure is presented in Algorithm \ref{companypredict}. 
\begin{algorithm}[ht]
    \caption{Prediction and uncertainty estimation algorithm of Company in UTC}
    \label{companypredict}
    \algrenewcommand\algorithmicrequire{\textbf{Input:}}
    \algrenewcommand\algorithmicensure{\textbf{Output:}}
    \begin{algorithmic}[1]
    \Require $\vecret[1:S][0:t]$: stock returns before $t$, $\{(\theta_n, w_n)\}_{n=1}^N$ : Traders
\Ensure $\mu[i][t+1]$: predicted return of stock $i$ at $t$, $\sigma^2[i][t+1]$: estimated variance of return of stock $i$ at $t$
    \Function{CompanyPrediction}{}
    \For{$n = 1,\cdots,N$}
        \State{$\mu_n$, $\sigma_n \Leftarrow f_{\theta_n,w_n}(\vecret[1:S][0:t])$
        }
        \Comment{Prediction by Trader \eqref{btctrader}}
    \EndFor
\State{$\mu = \frac{1}{N}\sum_{t \in \{1,2,\cdots,N\}} \mu_t$}
\State{$\sigma^2_{intra} = \frac{1}{N}\sum_{t \in \{1,2,\cdots,N\}} (\mu_t - \mu)^2$}
\State{$\sigma^2_{inter} = \frac{1}{N}\sum_{t \in \{1,2,\cdots,N\}} \sigma_t^2$}
\State{$\sigma = \sqrt{\sigma^2_{intra}+\sigma^2_{inter}}$}\\
\Return {$\mu, \sigma$}
    
    \EndFunction
    \end{algorithmic}
\end{algorithm}

Next, we explain how the Company optimizes Traders.
First, with respect to strategy selection, we perform an iterative process of pruning and generating, similar to the TC method.
Then, we optimize the weight distribution for each Trader in the UTC method.
We use MAP estimation since its solution can be analytically obtained because we assumed a simple normal distribution for the weights.
It helps reduce the computational cost of the UTC method because this optimization process is repeated for many traders.

\begin{algorithm}[ht]
    \caption{Educate algorithm of Company in UTC}
    \algrenewcommand\algorithmicrequire{\textbf{Input:}}
    \algrenewcommand\algorithmicensure{\textbf{Output:}}
    \label{companyeducate}
    \begin{algorithmic}[1]
    \Require $\vecret[1:S][0:t]$:stock returns before $t$
    \Require Traders. $N$ : the number of Traders. $Q$: ratio.
    \Ensure Traders
    \Function{CompanyEducate}{}
    \State{$R_n\Leftarrow R(f_{\Theta_n},\vecret[1:S][0:t],\ret[i][0:t+1])$}
    \Comment{Trader's return \eqref{traderloss}}

    \State{$R^* \Leftarrow$ bottom $Q$ percentile of $\{R_n \}$}
    \For{$n \in \{m| R_m \le R^* \}$}
        \Comment{for all bad traders}

        \State{Update traders' weight distributions by MAP estimation}
    \EndFor\\
    \Return{Traders}
    \EndFunction
    \end{algorithmic}
\end{algorithm}

Overall, they train the TC and UTC models as follows.
\begin{enumerate}
    \item Educate a fixed proportion of poorly performing Traders by Algorithm \ref{TCcompanyeducate} in TC and Algorithm \ref{companyeducate} in UTC.
    \item Replace a fixed proportion of poorly performing Traders with random new Traders by Algorithm \ref{TCcompanygen}.
    \item If the aggregation function $\mathrm{Aggregate}$ has trainable parameters, update them using the data $\vecret[1:S][t_1: t_2]$ and any optimization algorithm.
    \item Predict future returns by averaging in TC and Algorithm \ref{companypredict} in UTC.
\end{enumerate}

\subsection{Theoretical Properties}\label{sec:comp-tc}
Here we discuss the following properties of our UTC method.
\begin{description}
   \item[Proposition\ref{prop1}:]~\\ The UTC method can estimate the uncertainty~(posterior variance) of the prediction.
   \item[Proposition\ref{prop2}:]~\\The UTC method does not introduce additional biases from the TC method.
\end{description}

Traders and their empirical distribution are optimized to earn returns under the training data. 
Therefore, the empirical distribution of the Traders is expected to approximate the true posterior distribution given training data. 
The following proposition shows that the variance of our method approximates the posterior variance given the data under these assumptions.

\begin{proposition}[Posterior Variance Estimation]\label{prop1}
If the empirical distribution of the Trader $q(D)$ trained by UTC approximates the posterior distribution of the Trader $p(\theta|D)$ well, then the posterior variance of the return $\mathbb{V}_{p(y|x,D)}[y]$ can be approximated by the variance calculated by equation (\ref{eq:variance}).
\end{proposition}
\begin{proof}
Since we optimize the weight of the Trader by MAP estimation in Educate step, the Trader's output satisfies 
$$\mu_\theta(x) \approx \mathbb{E}_{p(y|x,\theta,D)}[y],\sigma^2_\theta(x) \approx \mathbb{V}_{p(y|x,\theta,D)}[y].$$
Furthermore, since the Trader $\{\theta_n\}$ is sampled from the distribution $
q(D) \approx p(y|x,D)$,
\begin{align*}
    & \mathbb{V}_{p(y|x,D)}[\mu_\theta(x)] \approx \frac{1}{N} \sum_{n} (\mu_n - \mu)^2 \\
    & \mathbb{E}_{p(y|x,D)}[\sigma^2_\theta(x)] \approx \frac{1}{N} \sum_{n} \sigma_n^2.
\end{align*}

From above, 
\begin{align*}
     &\mathbb{V}_{p(y|x,D)}[y] \\
     &= \mathbb{V}_{p(\theta|D)}\left[\mathbb{E}_{p(y|x,\theta,D)}[y]\right] + \mathbb{E}_{p(\theta|D)}\left[\mathbb{V}_{p(y|x,\theta,D)}[y]\right] \\
     &\approx \mathbb{V}_{p(y|x,D)}[\mu_\theta(x)] + \mathbb{E}_{p(y|x,D)}[\sigma^2_\theta(x)]\\
     &\approx \frac{1}{N} \sum_{n} (\mu_n - \mu)^2 + \frac{1}{N} \sum_{n} \sigma_n^2.
\end{align*}
\end{proof}
Next, we show that our method does not introduce additional biases from the TC method. 
That is, we show that the expected value of the output of our UTC method is identical with that of the TC method when the parameters are set to satisfy the assumptions of the following proposition. 
First, we show two lemmas about the unbiasedness of the Traders and Educate algorithm of the UTC.
\begin{lemma}[Unbiasedness of Prediction from TC] \label{lemma:prediction}
Let $C$ be the Company trained by the UTC method and $C'$ be the Company trained by the TC method. 
Let $\mathbb{E}[y]$ be the expected return of $C$ and $y'$ be the predicted return of $C'$.
If the parameters of each Traders $\{\theta_n\}_{n \in C}$ and $\{\theta_n\}_{n \in C'}$ is identical and expected weights of Traders $\{ \mathbb{E}[w_n] \}_{n \in  C}$ is identical with the weights of Traders $\{ w'_n \}_{n \in  C'}$, then, $\mathbb{E}[y] = y'$.
\end{lemma}
\begin{proof}
Let $\{z_n\}_{n \in C}$ and $\{z'_n\}_{n \in C'}$ be the signal of each Traders calculated by \eqref{eq:signal}.
Since, the parameters except weights is same, the signals satisfy $z_n = z'_n $.
From the definition of Trader, the predicted return by TC $\{y_n\}_{n \in C'}$ and our expected return $\{\mu_n\}_{n \in C}$ is calculated as 
\begin{align}
  & \mathbb{E}[y_n] = \mathbb{E}[w_n^\top z_n] = \mathbb{E}[w_n]^\top z_n, \nonumber\\
  & y'_n = w^\top z_n'. \nonumber
\end{align}
Therefore, $\mathbb{E}[y_n] = y'_n$.
If we use a simple averaging in Aggregate as well as TC method, then 
$$\mathbb{E}[y] = \frac{1}{N} \sum_{n} \mathbb{E}[y_n] = \frac{1}{N} \sum_{n} y'_n = y'$$.
\end{proof}
%
\begin{lemma}[Unbiasedness of Educate Algorithm] \label{lemma:educate}
Let $X,Y$ be the training data and $T$ be the Trader in UTC method and whose parameter be $\theta_{T}$ and weight distribution be $w$, $T'$ be the Trader in TC method and whose parameter be $\theta_{T'}$ and weight be $w'$,  If the parameters satisfy $\theta_{T} = \theta_{T'}$, then, the distribution of weights in UTC satisfies $\mathbb{E}[w] = w'$.
\end{lemma}
\begin{proof}
Let $Z$ and $Z'$ be the signal of each Trader.
Since, the parameters are identical, the signals satisfy $Z = Z' $.
We update the Trader's weight distribution to maximize the posterior distribution of weights given the data.
Suppose the prior distribution of weights is given as 
$$ w \sim \mathcal{N}(0,\sigma_0^2), T \sim \mathcal{N}(Y,\sigma^2). $$
Then, the posterior distribution is 
$$ p(w|T,Z) \propto \exp\left( \frac{-\|Zw-T\|^2}{2\sigma^2} - \frac{1}{2\sigma_0^2} w^\top w\right).$$
On the other hand, since the TC method updated weight by least square method, their weight is $$w' = (Z^\top Z + \lambda I)^{-1}Z^\top Y $$.

If we set the parameter of prior to satisfy $\lambda = \frac{\sigma^2}{\sigma_0^2}$, both weights satisfy $\mathbb{E}[w] = w'$.
\end{proof}
Finally, we prove the unbiasedness of the Company's prediction between UTC and TC when the parameters are set to satisfy the assumptions of these lemmas.
\begin{proposition}[Unbiasedness of Training and Prediction from TC]\label{prop2}
If the hyper-parameters of Company $C$ trained by the UTC method and Company $C'$ trained by the TC method satisfy the assumption of Lemma \ref{lemma:prediction} and \ref{lemma:educate}, the expected return of $C$ equals to the expected return of $C'$.
\end{proposition}
\begin{proof}
From Lemma \ref{lemma:educate}, The Traders' weights are consistent with respect to the expected value in the Educate process.
On the other hand, the TC method and our method are executed in the same way for the Prune-and-Generate algorithm. \\
Therefore, the distribution of Traders' parameters $\Theta$ in our method is identical to that of Traders' parameters $\Theta'$ in the TC method. Also, Traders' weights are also identical with respect to expected values. Therefore, from Lemma \ref{lemma:prediction}, the expected return of $C$ equals to the expected return of $C'$.
\end{proof}

\section{Experiment}
\subsection{Analysis on Synthetic Data} \label{sec:synthetic}
\subsubsection{Simple Nonlinear Case}
We used the following artificial data as a simple example of nonlinear time series.
\begin{align}
    & y_{0}(t) = 0.5 y_{0}(t-1) - 0.5y_{0}(t-1)y_{1}(t-1) \\
    & ~~~~~~~~~+ 0.1\min(y_{0}(t-1),y_{1}(t-1)) + \varepsilon_{0}(t) \nonumber \\
    & y_{1}(t) = -0.2 y_{1}(t-1) + 0.8y_{0}(t-1) \\
    & ~~~~~~~~~+ 0.5\max(y_{0}(t-1),y_{1}(t-1)) + \varepsilon_{1}(t) \nonumber
\end{align}
where $\varepsilon_{0}(t),\varepsilon_{1}(t)$ are generated by independently and identical normal distribution $N(0,0.1)$.

We used a simple Vector AutoRegression~(VAR~\cite{hamilton2020time}) model with lag 1 as a baseline.
Makridakis et al.\cite{makridakis2018statistical} demonstrated that traditional statistical methods such as the VAR model are more accurate than machine learning ones and suitable for the baseline of time series prediction tasks.
We used 1800 samples for training and used 200 samples for the test.
\begin{figure}[tb]
  \begin{minipage}[b]{0.5\textwidth}
    \centering
    \includegraphics[keepaspectratio, scale=0.5]{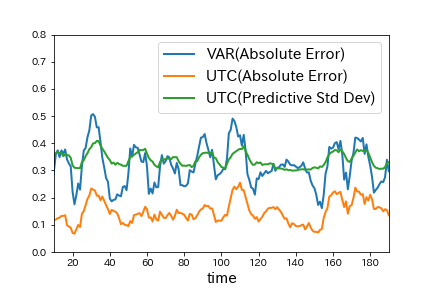}
  \end{minipage}
  \begin{minipage}[b]{0.5\textwidth}
    \centering
    \includegraphics[keepaspectratio, scale=0.5]{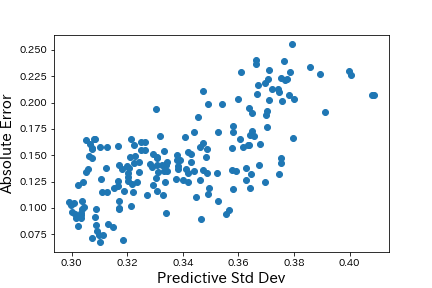}
  \end{minipage}
  \caption{Prediction error of synthetic data}
  \label{fig:synthetic}
\end{figure}

The upper panel of Figure \ref{fig:synthetic} shows the comparison of prediction error~(Absolute Error) between UTC and VAR. 
The lower panel shows the correlation diagram between the predictive standard deviation and the prediction error.

We can confirm that our UTC method makes better predictions than the VAR model because it captures nonlinear features of synthetic data. 
Also, we can observe that the estimated standard deviation is larger when the prediction error is large, i.e., prediction is difficult.
That means our UTC method can capture such uncertainty.

\subsubsection{Dataset Shift Case}
Next, we evaluated the performance of the UTC method in the dataset shift case using the artificial data defined as follows. 
\begin{align*}
    y_2(t) = \left\{ \begin{array}{ll}
        y_0(t-1)+y_1(t-1) + \epsilon_2(t) &  (t < 200) \\
        y_0(t-1)-y_1(t-1) + \epsilon_2(t) & (t \geq 200)
    \end{array} \right.
\end{align*} 
where $\varepsilon_{2}(t)$ are generated by independently and identical normal distribution $N(0,0.1)$.
To verify the adaptability of our method to dataset shift, we sequentially updated the model using the last 100 data as training data.
Figure \ref{fig:datashift} shows the result. 
In this experiment, the error increases as data distribution are shifted (Blue region). Along with this, the estimated variance is also increasing. Furthermore, as the environment shifts to the original static environment, the error and the estimated variance decrease again.
This means that our method detects the decrease in reliability of prediction due to the dataset shift.
\begin{figure}[tb]
  \begin{minipage}[b]{0.5\textwidth}
    \centering
    \includegraphics[keepaspectratio, scale=0.3]{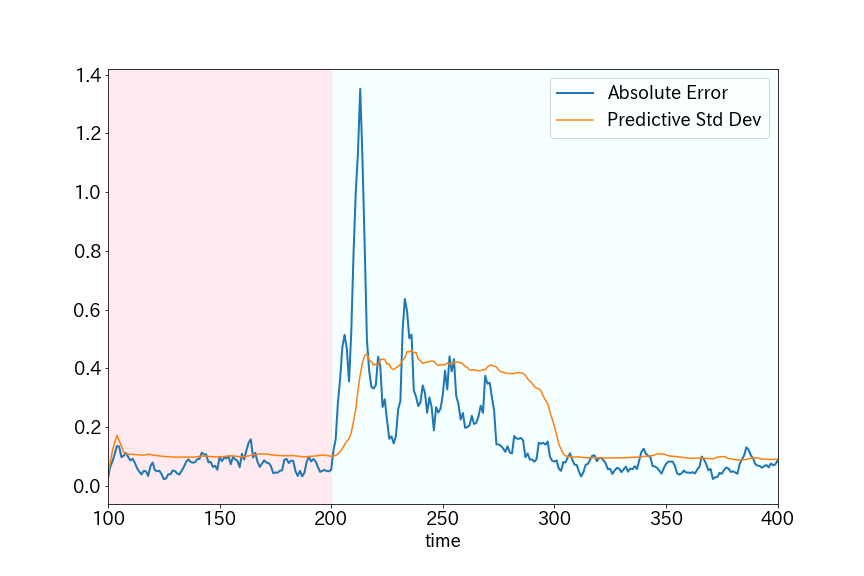}
  \end{minipage}
  \caption{Prediction error and predictive standard error in dataset shift}
  \label{fig:datashift}
\end{figure}
\subsection{Experiments on Real Datasets}\label{sec:exp}
\subsubsection{Dataset}
We tested two different settings using real market data, the TOPIX100 index.
The TOPIX100 Index is a market capitalization-weighted index of large-cap Japanese stocks, consisting of the top 100 stocks with particularly high market capitalization and liquidity (trading value) among the stocks included in the TOPIX index.

We performed daily trading at the close and hourly intraday trading. 

In the daily trading setting, we used daily closing prices of the constituents of the TOPIX100 index from January 4, 2010, to January 31, 2022.
We used the data form from January 4, 2010, to December 30, 2017, for training and from January 4, 2018, to January 31, 2022, for the test.

In the intraday trading setting, we used hourly closing prices of the constituents of the TOPIX100 index from January 4, 2016, to June 30, 2021.
We used the data form from January 4, 2016, to December 30, 2019, for training and from January 4, 2020, to June 30, 2021, for the test.

\subsubsection{Experimental Settings}
The following settings are the same as \cite{ito2021trader}.
We introduced a time window $w>0$ and a trading execution lag $l>0$ as in \cite{ito2021trader}.
Throughout experiments, we used $w= 10$ and $l = 1$.
We trained models using observations $\vecret[1:S][t-l-w:t-l]$ and predict returns at every time $t$.

Recall $r_i[t]$ be the return of stock $i$ ($i \in \{ 1, \ldots, S \}$) at time $t$. 
Here, the time between $t$ and $t+1$ represents 1 day in the case of daily trading, and 1 hour in the case of intraday trading.

To evaluate the effectiveness of our UTC method, we performed the following virtual trading strategy. 
We buy the stock if the predicted label for the time $t$ (indicating whether the price of the time $t+1$ would rise or fall) is positive and sell it otherwise.
We define the strategy's return at time $t$ as 
$$R[t] = \mean(\hat{b}_i[t]\times r_i[t])$$
where $\hat{r}_i[t]$ is its prediction from each method, $\hat{b}_i[t] = \mathrm{sign}(\hat{r}_i[t])$ and $\mean(\cdot)$ represents the average value of the function over its input.
We compared our UTC methods with the following strategies:
\begin{itemize}
    \item Market: Simply buy all stocks in TOPIX100 index equally~\cite{demiguel2009optimal} i.e., $R[t] = \mean(r_i[t])$.
    \item VAR: $\hat{r}_i[t]$ is estimated by the VAR model using historical observations. 
    We used the VAR(1) model selected by the AIC for lags from 1 to 10.
    \item RF: $\hat{r}_i[t]$ is estimated by Random Forest model with Scikit-learn package~\cite{pedregosa2011scikit}.
    We set "n\_estimators" to 100, "min\_samples\_split" to 10, "min\_samples\_leaf" to 4, "max\_features" to sqrt, "max\_depth" to 60.
    We determined these hyper-parameters by cross-validation.
    \item TC: $\hat{r}_i[t]$ is estimated by the original TC method.
    \item UTC: Let $\sigma_i[t]$ be its prediction uncertainty from UTC method. Define $\hat{b}_i[t] = \mathrm{sign}(\hat{r}_i[t] \times I_{A})$ where $I$ is an indicator function and $A$ is an event when $\sigma_i[t]$ is higher than the threshold calculated from past predictive variance. 
    Then $R[t] = \mean(\hat{b}_i[t]r_i[t])$.
\end{itemize}

Table \ref{tab:hyperparamters} lists the hyper-parameters used in both the TC and UTC methods.
\begin{table}[tb]
\caption{Hyper-parameters used in TC and UTC in the experiment}
\label{tab:hyperparamters}
\centering
\begin{tabular}{ccc}
    \hline
    \textbf{Parameter} & \multicolumn{1}{c}{\textbf{Value}} & \textbf{Def} \\
    $\numterm$ &  $\{1,\cdots,10\}$  & Definition \eqref{TCtraders}\\
    $D_j,F_j$ & $\{0,\cdots,10\}$ & Definition \eqref{TCtraders}\\
    $A_j(x) $ &  $\{x,\mathrm{tanh}(x),\mathrm{exp}(x),\mathrm{sign}(x),\mathrm{ReLU}(x)\}$ & Definition \eqref{TCtraders}\\
    $O_j(x,y)$ & \begin{tabular}{@{}c@{}}$\{x+y, x-y , xy, x, y,\max(x,y)$, \\ $ \min(x,y),x>y,x<y,\textrm{Corr}(x,y) \}$\end{tabular}  & Definition \eqref{TCtraders}\\
    $N$ & $200$ & Algorithm \ref{companypredict} \\
    Aggregate & Simple Averaging & Algorithm \ref{companypredict} \\
    $Q$ & $0.1$ & Algorithm \ref{TCcompanygen} \\
    \hline
\end{tabular}
\end{table}

\subsubsection{Performance Measures}
To evaluate the performances of each strategy, we adopted the four metrics widely used in financial experiment\cite{brandt2010portfolio} as follows. 
$T_Y$ represents the number of periods of trading in one year.
We used $T_Y = 250$ in the daily setting, $T_Y = 1250$ in the hourly setting.
\begin{itemize}
    \item \textbf{Annualized Return(AR)}: 
    We define the Annualized Return \textbf{(AR)} as 
        \begin{align}
        \mathrm{AR} := T_Y\times \mean(R[t])
        \end{align}
    This measure represents the profitability of the strategy.
    \item \textbf{Annualized Risk(RISK)}:We define the Annualized Risk \textbf{(RISK)} as
        \begin{align}
         \mathrm{RISK} := \sqrt{T_Y} \times \sqrt{\sum_{t=1}^T(R[t]-\mean(R[t]))^2/(T-1)}.
        \end{align}
    \item \textbf{Sharpe Ratio (SR)}: The Sharpe ratio \cite{sharpe1964capital}, or the Return/Risk ratio (R/R) is the return adjusted by its risk. 
    That is, 
        \begin{align}
        \mathrm{SR} := \mathrm{AR} / \mathrm{RISK}.
        \end{align}
    \item \textbf{Maximum DrawDown (MDD)}: The Maximum DrawDown is defined as the largest drop from an extremum~\cite{magdon2004maximum}:
        \begin{align}
        \mathrm{MDD} := \max_{1 \leq t \leq T}\max_{t < s \leq T}(1 - C[t]/C[s])
        \end{align}
    where $C[t] := \sum_{j=1}^t R[j]$ is the cumulative return.
    \item \textbf{Cramer Ratio(CR)}: The CR are adjusted returns by its MDD~\cite{young1991calmar}. CR is more sensitive to drawdown events that occur less frequently:
        \begin{align}
        \mathrm{CR} := \mathrm{AR} / \mathrm{MDD}
        \end{align}
\end{itemize}

Larger AR, SR, and CR are better, while smaller RISK and MDD are better.

\subsubsection{Result}

\begin{table}[t]
\caption{Performance Comparison in the daily trading setting: Average results of 5 repeated experiments}
\label{tab:jpn}
\begin{center}
\begin{tabular}{cccccc}\hline\hline
     & UTC & TC & Market & RF & VAR \\ \hline
AR  &   5.42  & \textbf{9.03} & 8.56 & 7.80 & 8.13 \\
RISK & \textbf{4.45} & 9.05 & 17.90 & 8.37 & 16.90   \\
SR  & \textbf{1.22}   & 1.00 & 0.48 & 0.93 & -0.48\\
MDD  & \textbf{-6.43} & -13.99 & -28.66 & -12.12 & -27.03\\ 
CR & \textbf{0.86} & 0.64 & 0.30 & 0.64 & 0.30 \\
\hline
\end{tabular}
\end{center}

\end{table}
\begin{figure}[tb]
    \centering
    \includegraphics[width=\columnwidth]{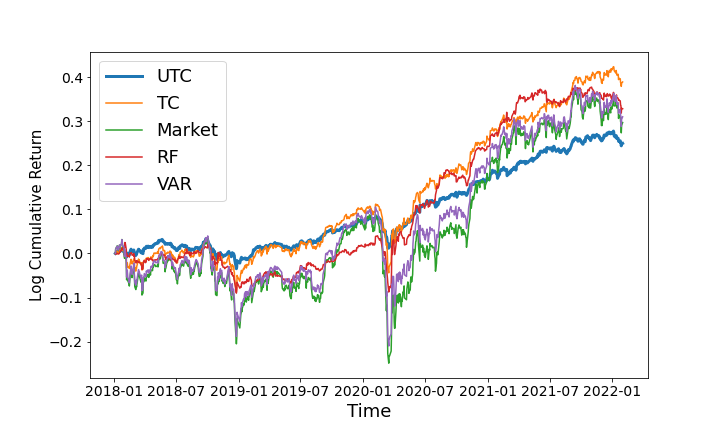}
    \caption{Cumulative daily returns on JPN market}
    \label{fig:offline_jpn}
\end{figure} 

Table \ref{tab:jpn} shows the comparison between our proposed method (UTC) and other baselines in the daily trading. We ran the experiment 5 times and averaged the result.
Figure \ref{fig:offline_jpn} shows the transition of cumulative return $C[t]$ of each method.
As previous studies have shown, the TC method was confirmed to earn higher returns than other methods in both daily and intraday trading settings.
However, TC also had a high risk because it does not consider uncertainty.
On the other hand, our method invests only in stocks with low predictive uncertainty, thereby achieving stable returns while limiting risk. Therefore, our method achieved the highest SR and CR, risk-adjusted performance evaluation metrics.
In particular, during the COVID-19 shock in the market in the first half of 2020, our method successfully captured market uncertainties and limited investments during those periods and significantly suppressed the sharp fall.

Table \ref{tab:jpn_hour} shows the comparison between our proposed method (UTC) and other baselines in the intraday trading. We ran the experiment 5 times and averaged the result.
We also show the transition of cumulative return $C[t]$ of each method in  Figure \ref{fig:offline_jpn_hour}.
As in the daily case, our UTC method also achieved the highest SR and CR in the intraday case.

\begin{table}[t]
\caption{Performance Comparison in the hourly trading setting: Average results of 5 repeated experiments}
\label{tab:jpn_hour}
\begin{center}
\begin{tabular}{cccccc}\hline\hline
     & UTC & TC & Market & RF & VAR \\ \hline
AR  & 9.60 & \textbf{11.9} & 7.75 & 2.70 & 4.11 \\
RISK & 6.42 & 8.42 & 21.2 & 11.62 & \textbf{5.06}   \\
SR  & \textbf{1.52} & 1.46 & 0.37 & 0.23 & 0.81\\
MDD  & -6.95 & -9.83 & -30.96 & -15.87 & \textbf{-4.57}\\ 
CR & \textbf{1.52} & 1.37 & 0.25 & 0.17 & 0.90 \\
\hline
\end{tabular}
\end{center}
\end{table}
  %
\begin{figure}[t]
    \centering
    \includegraphics[width=\columnwidth]{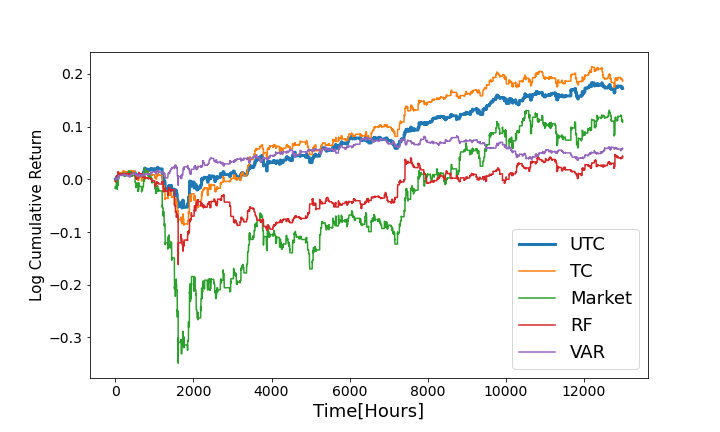}
    \caption{Cumulative hourly returns on JPN market}
    \label{fig:offline_jpn_hour}
\end{figure}

\section{Related Work}
\subsection{Stock Return Prediction in Finance}

Predicting future stock prices has been actively researched for a long time but is still a highly challenging task~\cite{chen2018integrating}.
Stock price prediction can be broadly divided into two types: fundamental analysis and technical analysis\cite{atsalakis2009,sedighi2019novel}.
Fundamental analysis focuses on fundamental information of the corporation, such as a company’s revenues and expenses, yearly growth rate, and other information contained in financial statements.
On the other hand, the technical analysis predicts using market data such as historical stock price and volume data.

The methods of the former analysis perform a regression analysis using cross-sectional data of fundamental information.
Such strategies aim to build a portfolio for investing as a subset of a large bucket of stocks. These approaches are frequently applied to a practical quantitative investment strategy~\cite{nakagawa2020ric}.
One of the most significant interests in a cross-sectional analysis lies in finding ``factors'' that have strong predictive powers to the stock return of portfolio strategy~\cite{fama1992}.
This argument inspires a lot of subsequent studies that propose more sophisticated versions of factors~\cite{harvey2016}. 

The methods of the latter strategy analyze past stock prices as time-series data~\cite{atsalakis2009} and are applied to a practical trading strategy that focuses on a particular stock. 
The autoregressive integrated moving average (ARIMA) model, Vector autoregression (VAR), and generalized autoregressive conditional heteroscedasticity (GARCH~\cite{bollerslev1986generalized}) model are basic benchmarks often used in financial time series prediction. These linear models take into account uncertainties as error terms, and some studies have also modeled their distributions~\cite{Ledolter1979}.

\subsection{Application of Machine Learning to Stock Return Prediction}
With the introduction of artificial intelligence and machine learning, these techniques have received increased attention in stock prediction studies\cite{atsalakis2009,cavalcante2016computational}. 
Unlike traditional time series methods, these methods can handle the nonlinear, noisy, and complex data of the stock market, leading to more effective predictions~\cite{chen2018integrating,ito2021trader}. 
Among them, Trader-Company~(TC) method achieves state-of-the-art performance in this field~\cite{ito2021trader}.
The TC method takes into account the dynamism of the stock market and has both high predictive power and interpretability.
The TC method is based on a meta-heuristic approach, which has been extensively applied to stock prediction~\cite{soler2017survey,sedighi2019novel}.

Another promising approach is the application of neural-network-based models such as CNNs and RNNs.
They have been proposed as innovative and advantageous alternatives to traditional methods~\cite{hu2021,Lai2018,ijcai2020-640,imajo2021deep}.

These studies showed promising performance, but unfortunately, they lack a perspective on the uncertainty of the prediction.
In other words, all of these methods are only focused on point prediction. This is a major practical problem because of the possibility of significant losses in investment strategies.
As noted in the previous subsection, research in the field of finance has incorporated uncertainty, but in terms of adapting machine learning to finance, there remains room for improvement. Among them, several applications of Bayesian neural networks have been proposed to deal with the issue. These studies show their effectiveness in highly volatile and uncertain markets while COVID-19 is reported~\cite{chandra2021}. 
Inspired by these studies, this study introduced the uncertainty estimation technique to the TC method to capture uncertainty.

\subsection{Uncertainty Estimation with Weight Parameterization}
Most methods for uncertainty estimation have been proposed under a Bayesian framework. Bayesian estimation requires some approximation to estimate the uncertainty unless the linear model since exact computation is usually intractable. 
One basic approach to the approximation is to set parameters corresponding to the mean and variance in the model and use these to approximate the posterior distribution. 
In linear regression, Bayesian regression follows this approach, and methods that follow this approach have been proposed for other standard models, such as neural networks~\cite{pmlr-v37-blundell15} and decision trees~\cite{pmlr-v119-duan20a}. 
The former assumes the Gaussian distribution for weights in neural networks, and the latter assumes the outputs follow Gaussian distributions and predict their mean and variance. Both works train the model to minimize the KL divergence from the posterior distribution. Our method introduces these approaches in the estimation of the uncertainty in the Trader. This approach is computationally efficient, but it has been pointed out that these methods tend to make predictions biased toward one specific mode ~\cite{Fort2019,malinin2021uncertainty}.
 
 \subsection{Uncertainty Estimation with Ensemble}
Another approach is ensemble-based, which approximates the posterior distribution using multiple outputs. For a neural network, Gal and Ghahramani proposed Monte Carlo Dropout~\cite{pmlr-v48-gal16}, which applies dropout during inference and estimates uncertainty. This method can be interpreted as using multiple sub-networks that share the parameters. 
Deep Ensemble~\cite{NIPS2017_9ef2ed4b,d'angelo2021repulsive} trains several neural networks with different initialization and estimates the predictive uncertainty by the variance of their outputs. Learning independent networks makes them less likely to be biased toward one particular mode and enhances their robustness~\cite{fort2020deep}.
For tree-based models, initial work trained using MCMC~\cite{NIPS2006_1706f191} and the combination with GBDT~\cite{malinin2021uncertainty} shows better performance in terms of both prediction and uncertainty estimation.
These methods are experimentally confirmed to be more robust to dataset shift than methods that explicitly learn distributions\cite{NIPS2017_9ef2ed4b,malinin2021uncertainty}. These characteristics are effective in non-stational environments such as financial markets.
Our method introduces these approaches for the Company to capture the uncertainty due to the dataset shift.

\section{Conclusion}
We proposed a novel approach called the Uncertainty Aware Trader-Company Method~(UTC), which extends the TC method with probabilistic modeling framework.
The UTC method can keep the predictive power and interpretability of the TC method while capturing uncertainty.

Our contributions are as follows:
\begin{itemize}
    \item We showed that the UTC method could estimate the prediction's uncertainty (the posterior variance) under the assumption that the empirical distribution is a good approximation of the posterior distribution.
    \item We also proved that the UTC method does not introduce additional biases from the TC method.
    \item We confirmed on synthetic data that the UTC method could detect the situations where the prediction is difficult or the data distribution is abruptly changed.
    \item We demonstrated on real market data that the UTC method could achieve higher risk-adjusted returns than baselines.
\end{itemize}

For the direction of further study, we may consider uncertainty when recruiting or firing Traders in the Prune-and-Generate algorithm.
Furthermore, our method assumed the distribution of the parameters is the Gaussian mixture distribution. We may consider modeling them in a non-parametric way.

\end{document}